\documentclass[11pt]{article}
\usepackage{hyperref}
\usepackage{amsmath}
\usepackage{amssymb}
\usepackage{amsthm}
\usepackage{enumitem}
\usepackage{graphicx}
\usepackage{subcaption}
\usepackage{physics}
\usepackage{biblatex}
\usepackage[margin=1in]{geometry}
\usepackage{tikz}
\usepackage{adjustbox}
\usepackage{todonotes}
\usetikzlibrary{quantikz2}
\usetikzlibrary{arrows}

\newtheorem{theorem}{Theorem}[section]

\newtheorem{fact}[theorem]{Fact}

\newtheorem{lemma}[theorem]{Lemma}

\newtheorem{definition}[theorem]{Definition} 

\theoremstyle{remark}
\newtheorem{remark}{Remark}[section]

\numberwithin{equation}{section}

\newcommand\numberthis{\addtocounter{equation}{1}\tag{\theequation}}
\newcommand{\zz}{\mathbb{Z}}
\DeclareMathOperator{\poly}{poly}
\DeclareMathOperator{\pr}{Pr}
\DeclareMathOperator{\ex}{\mathbb{E}}

\newcommand{\ii}{\mathfrak i}
\renewcommand{\var}{\text{Var}}

\usepackage{authblk}
\author[1]{Jin-Yi Cai}
\author[1,2]{Ben Young}
\affil[1]{University of Wisconsin-Madison, USA}
\affil[2]{IT University of Copenhagen, Denmark}

\title{A
Sharp Noise Threshold for Shor’s Quantum Factoring and
Discrete Log Algorithms}
\date{}

\begin{document}
\maketitle

\begin{abstract}
    We study the asymptotic behavior of Shor's quantum factoring and discrete log
algorithms when noise affects the precise controlled rotation gates in their
quantum Fourier transforms. Improving the results of Cai (2024) and Cai and
Young (2025), we identify a sharp and vanishingly small noise threshold. If the noise level lies
below this threshold, then the two algorithms succeed in expected polynomial
time. If the noise level exceeds this threshold, then the algorithms provably fail to solve their 
respective problems in expected polynomial time when the underlying primes belong to a set
of positive density.
\end{abstract}

\section{Introduction}
Shor's quantum algorithms for the factoring and discrete log problems are
arguably the two most famous results in quantum computing.
The \emph{factoring problem} (restrict to two prime factors) is defined as follows:
given an integer $N$ promised
to be the product of unknown primes $P$ and $Q$, find the prime factors.
The \emph{discrete log problem} is defined as follows: given a prime number $P$, 
a nonzero natural number $g < P$ such that $g^0,g^1,g^2,\ldots,g^{P-2}$ generate all nonzero integers 
mod $P$, and an integer $y$ which is nonzero mod $P$,
find the \emph{discrete log value}, which is an integer $0 \leq d \leq P-2$ such that $g^d \equiv y \bmod{P}$.
The assumed hardness of these problems underly the RSA cryptosystem and Diffie-Hellman key 
exchange, respectively, two famous cryptographic protocols. 

The key step in Shor's algorithms for the factoring and discrete log problems is the quantum Fourier transform (QFT).
The quantum circuit implementing the QFT is composed mostly of rotations about the $Z$ axis by angles of exponentially
decreasing magnitude. Specifically, the body of the $n$-qubit
QFT circuit is made up of
controlled-$R_{2^{-k}}$ gates for $k = 1,\ldots,n$, where $n$ is on the order of the problem input size (e.g.~the number
of bits in the integer to be factored) and
\[
    R_{2^{-k}} = \begin{bmatrix} 1 & 0 \\ 0 & e^{2\pi \ii 2^{-k}} \end{bmatrix}
\]
is a rotation by angle $2\pi/2^k$
(throughout, $\ii$ denotes the imaginary unit).
The finest rotation is by angle $2\pi/2^n$, a quantity exponentially small in the input size.
Several previous works have examined these exponentially small rotations and their effect on the performance
of the QFT. Coppersmith \cite{coppersmith} introduced
a \emph{banded} QFT circuit that omits all
controlled-$R_{2^{-k}}$ gates for $k$ larger than a parameter $b$ on the order of $\log(n)$.
Fowler and Hollenberg \cite{fowler_scalability_2004, fowler_erratum_2007} show that small controlled-$R_{2^{-k}}$ gates
are difficult to implement fault-tolerantly and provide further analysis of Coppersmith's model. 
Nam and Bl\"umel \cite{nam_robustness_2014, nam_performance_2015,  nam_structural_2015} reintroduce the small controlled-$R_{2^{-k}}$ gates,
but assume that all rotation gates are subject to errors, including, as studied in this work, Gaussian random perturbations
of the rotation angle. These papers claim, using a mix of numerical simulation for small $n$ and heuristic analysis for general $n$, that
the QFT, and hence Shor's factoring algorithm, remains effective on practical scales -- including the scale of integers used in modern RSA
implementations -- for small but reasonable noise levels and $b$ values. However, a heuristic estimate in \cite[Equation 70]{nam_robustness_2014} suggests that, for a fixed noise level or even noise levels vanishingly small in $n$, the performance of the QFT decays exponentially in $n$. Therefore, without the ability to correct arbitrarily small errors, algorithms based on the QFT will fail on sufficiently large inputs.\footnote{We discuss briefly the effect of Quantum Error Correction at the end of the paper; see \autoref{rem:qec}.}

Moving beyond the heuristic estimates in \cite{nam_robustness_2014, nam_performance_2015,  nam_structural_2015},
the first author \cite{cai} gave the first rigorous proof
under any error model that Shor's factoring algorithm fails when noise is present on sufficiently large inputs.
He used a \emph{relative} error
model similar to the relative
error models in \cite{barenco_approximate,nam_robustness_2014,nam_performance_2015}, but allowing exact rotations for large angles: Keep each controlled-$R_{2^{-k}}$ 
gate exact for all $k < b$; for $k \geq b$, replace each controlled-$R_{2^{-k}}$ 
gate in the QFT by a noisy version
\[
    \widetilde{R}_{2^{-k}} = \begin{bmatrix} 1 & 0 \\ 0 & e^{2\pi \ii 2^{-k}(1+\epsilon r)} \end{bmatrix}
\]
which rotates about $Z$ by a corrupted angle
$2\pi(1+\epsilon r)/2^k$, where $\epsilon$ is a global noise level parameter and $r$ is a Gaussian
random variable $\sim N(0,1)$. Here, we say the noise is ``relative'' because it is scaled by the size of each gate's rotation.
The noise variables $r$ on each gate are independent and
the same random perturbation on a
gate applies to every state in a superposition to which the gate is applied. 
Under this error model, the first author proves that, if 
$b + \log_2(1/\epsilon) < \frac{1}{3} \log_2(n) - c$ for some constant $c > 0$, or
equivalently if
\begin{equation}
    \frac{\epsilon}{2^b} = \Omega\left(\left(\frac{1}{n}\right)^{1/3}\right),
    \label{eq:cai}
\end{equation}
then Shor's factoring algorithm fails to factor
$n$-bit integers $PQ$, for random primes  $P,Q$  of the appropriate size, or taken from an explicitly defined set of primes of positive density. 
Under the same error model, the authors \cite{discrete_log} proved a similar result for
Shor's discrete log algorithm (along with supporting experimental results): if there exists some constant $0 < c < 1/2$ such that
$b + \log_2(1/\epsilon) < \frac{1-c}{2} \log_2(n) - \Theta(1)$ 
for large $n$, or equivalently
\begin{equation}
    \frac{\epsilon}{2^b} = \Omega\left(\left(\frac{1}{n}\right)^{\frac{1-c}{2}}\right),
    \label{eq:cai_young}
\end{equation}
then Shor's discrete log algorithm, for a positive density of
primes $P$ of binary length $n$, fails to find the discrete log modulo $P$ of all but an 
exponentially small fraction of inputs $y \in \zz_P^*$.


\paragraph{Results.}
In this article, we improve our above results to a sharp noise threshold for both the factoring
(\autoref{thm:factor}) and discrete log (\autoref{thm:dl}) problems. 
Let $n$ be the input length (the bitlength of $N$ or $P$ for the factoring or discrete log
problems, respectively), and suppose the controlled rotation gates in the $n$-qubit
QFT are subject to the error model from \cite{cai,discrete_log}. 
We show that, if the noise level is bounded by
\begin{equation} \label{eq:lower_bound}
    \frac{\epsilon}{2^b} = O\left(\left(\frac{\log n}{n}\right)^{1/2}\right),
\end{equation}
then Shor's algorithm solves the factoring and discrete log problems in expected time polynomial in $n$ under this error model (however, see \autoref{rem:qec} below).
On the other hand, if
\begin{equation}
    \lim_{n\to\infty} \left(\frac{\epsilon}{2^b}\right) \Big/  \left(\frac{\log n}{n}\right)^{1/2}
    = \infty,
    \label{eq:our_bound}
\end{equation}
then there is a set $S$ of primes of positive density such that, for primes $P,Q \in S$ or
$P \in S$, respectively,
Shor's algorithm, with probability exponentially close to 1,
fails to solve the factoring and discrete log problems in time polynomial in $n$.
Note that our bound \eqref{eq:our_bound} is tighter than the bound \eqref{eq:cai}
and slightly tighter than \eqref{eq:cai_young}.
Our positive result furthermore shows that \eqref{eq:our_bound} is asymptotically tight.



In both of Shor's algorithms, the QFT creates a state whose
coefficients are sums of points on the unit circle.
In the noise-free case, the points in the coefficients of certain
desired states are concentrated in a segment
of the unit circle. Hence the probability of measuring one
of these desired states is large, and the factors of $N$ or the
discrete log value are obtainable from such a state via
classical postprocessing. However,
noise in the QFT perturbs the points within each coefficient; 
if the noise is above a certain threshold,
the points spread evenly around the unit circle and their sum,
hence the probability of measuring a desired state, is
small. The amount of noise affecting the distribution of points
in each coefficient depends on the distribution of 1 bits
in the binary representations of certain problem parameters;
we apply the bit distribution analysis of \cite{cai,discrete_log}, then
give a sharper analysis of the sum of points on the unit circle, giving our improved superpolynomially small upper bound 
\eqref{eq:our_bound} on the success probability upper bound in 
the high-noise case.
The inverse-polynomial lower bound \eqref{eq:lower_bound} on the success probability 
in the low-noise case for factoring requires a
sharpening of Shor's analysis;
for discrete log, the lower bound further requires an additional
bit distribution argument.

For fixed $b$ and $\epsilon$, a sufficiently large $n$ will eventually satisfy \eqref{eq:our_bound},
hence the algorithms fail on large classes of inputs of size $n$.
We emphasize that our analysis here does not rule out the possibility that
actual quantum computers in the future can efficiently factor integers
and compute discrete logs on a scale of interest to  modern cryptosystems.
Indeed, our results show exactly how much error in the QFT the algorithms can tolerate before
they fail (under the noise model and if that is the only source of noise). However, we also prove that, if the algorithms are to solve these problems
for \emph{arbitrarily large} inputs, their QFTs must use \emph{arbitrarily precise} 
controlled rotations, and at exactly what precision.
Hence, to conclude from the existence of these algorithms that the factoring and discrete log problems
are polynomial-time tractable on a quantum computer,
one must assume that arbitrarily precise controlled rotations are physically realizable. See \autoref{rem:qec} for further
discussion.

\section{Background}
\subsection{The noisy quantum Fourier transform}
\label{sec:dlognoise}
For integer $x > 0$, define $[x] := \{0,1,\ldots,x-1\}$, and for integers 
$0 \leq x < y$, define $[x,y] := \{x,\ldots,y\}$ and $[x,y) := \{x,\ldots,y-1\}$.
We encode integers $0 \leq x < 2^n$ as $n$-qubit quantum states $\ket{x} = \ket{x^{[n-1]}\ldots x^{[1]}
x^{[0]}} \in (\mathbb{C}^2)^{\otimes n}$, where $x^{[j]}$ is the value of the $j$th least significant bit in 
the $n$-bit binary representation of $x$ -- that is, $x =  2^{n-1} x^{[n-1]} + \ldots + 2 x^{[1]} + x^{[0]}$.
The exact $n$-qubit QFT $F_{2^n}$ sends state $\ket{x}$ to 
\begin{equation}
    F_{2^n}\ket{x} = \frac{1}{2^{n/2}} \sum_{v=0}^{2^n-1} \exp(2\pi \ii \frac{xv}{2^n})\ket{v}.
    \label{eq:nqubitfourier}
\end{equation}
This power-of-two QFT $F_{2^n}$ is easily implemented using a quantum circuit composed 
of Hadamard gates and controlled-$R_{2^{-k}}$ gates for $2 \leq k \leq n$, 
where $R_\theta$ is the single-qubit rotation about $Z$ by angle $2\pi \theta$: 
\[
    R_\theta = \begin{bmatrix} 1 & 0 \\ 0 & e^{2\pi \ii\theta} \end{bmatrix}.
\]
See \cite[Section 5.1]{nielsen_chuang} for an explicit description of the circuit 
implementing $F_{2^n}$.
We consider the scenario where there is some $b < n$ such that for every $R_{2^{-k}}$ gate with $k \geq b$ (which are
small rotations), some noise is unavoidable. We model this as a random relative error of 
a small magnitude.
Specifically, we replace each controlled-$R_{2^{-k}}$ gate in the circuit implementing 
$F_{2^n}$ by a controlled-$\widetilde{R_{2^{-k}}}$ gate, where
\[
    \widetilde{R_\theta} = \begin{bmatrix} 1 & 0 \\ 0 & e^{2\pi \ii(1+\epsilon r) \theta} 
    \end{bmatrix}
\]
is a rotation about $Z$ of angle $2\pi(1+\epsilon r) \theta$, where $r$ is an independent Gaussian random variable drawn from $N(0,1)$ and $\epsilon$ is a global noise magnitude parameter.
With the exact rotations $R_{2^{-k}}$, the $F_{2^n}$ circuit (including the qubit order
reversal at the end) effects the transformation
\begin{align*}
    \ket{x} \mapsto \frac{1}{2^{n/2}} 
     &\left(\ket{0} + \exp(2\pi \ii 0.x^{[0]})\ket{1}\right) \otimes \\
                    & \vdots \\
     &\left(\ket{0} + \exp(2\pi \ii 0.x^{[n-2]}\ldots x^{[0]})\ket{1}\right)  \otimes \\
    &\left(\ket{0} + \exp(2\pi \ii 0.x^{[n-1]}x^{[n-2]}\ldots x^{[0]})\ket{1}\right),
    \numberthis\label{eq:equiv_fourier}
\end{align*}
an equivalent expression for the state $F_{2^n} \ket{x}$ in \eqref{eq:nqubitfourier} 
(see \cite[Section 5.1]{nielsen_chuang}; note that, relative to that exposition, we reverse 
the bit indexing of $\ket{x}$ in both \eqref{eq:nqubitfourier} and \eqref{eq:equiv_fourier}).
When each controlled-$R_{2^{-k}}$ gate is replaced by a 
controlled-$\widetilde{R_{2^{-k}}}$ gate for $k \geq b$, the noisy circuit implements a transformation we call
$\widetilde{F_{2^n}}$, where
\begin{align*}
    \widetilde{F_{2^n}} \ket{x} = 
    &\Big(\ket{0} + \exp(2\pi \ii 0.x^{[0]}) \ket{1}\Big)  \otimes\\
    &\quad\vdots \\
    &\left(\ket{0} + \exp\left(2\pi \ii 0.x^{[b-2]}\ldots x^{[0]}
    \right)
    \ket{1} \right)  \otimes\\
    &\left(\ket{0} + \exp\left(2\pi \ii \left(0.x^{[b-1]}\ldots x^{[0]} + \frac{\epsilon}{2^b}
    r_0^{(n-b)} x^{[0]}
    \right)\right)
    \ket{1} \right)  \otimes\\
    &\quad\vdots \\
    &\left(\ket{0} + \exp\left(2\pi \ii \left(0.x^{[n-2]}\ldots x^{[0]} + \frac{\epsilon}{2^b}
    \left[\frac{r_0^{(1)} x^{[n-b-1]}}{2^0} + \ldots + \frac{r_{n-b-1}^{(1)} x^{[0]}}{2^{n-b-1}}\right]
    \right)\right)
    \ket{1} \right) \otimes \\
    &\left(\ket{0} + \exp\left(2\pi \ii \left(0.x^{[n-1]}x^{[n-2]}\ldots x^{[0]} + \frac{\epsilon}{2^b}
    \left[\frac{r_0^{(0)} x^{[n-b]}}{2^0} + \ldots + \frac{r_{n-b}^{(0)} x^{[0]}}{2^{n-b}}\right]
    \right)\right)
    \ket{1} \right).
\end{align*}
All of our results will use the following fact quantifying the corrupting effect of
Gaussian noise on the cosine of an angle.
\begin{fact}[\cite{cai}] \label{fact:ex}
    For Gaussian random variable $X$ centered at zero and angle $\theta$,
    \[
        \ex[\cos(\theta + X)] = \cos(\theta) \exp(-\frac{1}{2} \emph{Var}[X]).
    \]
\end{fact}
\begin{proof}
Recall the expectation of the cosine of a Gaussian random variable centered at zero 
\cite{cai}: if $X \sim N(0,\sigma^2)$, then the expectation
\[
    \ex[\cos(X)] = e^{-\sigma^2/2} \text{ and } \ex[\sin(X)] = 0
\]
(with the second statement using that $\sin$ is an odd function). Then
\[
    \ex\left[\cos\left(\theta + X\right)\right]
    = \ex\left[\cos\left(\theta\right)\cos(X) 
    - \sin\left(\theta\right)
    \sin(X)\right] 
    = \cos\left(\theta\right)\exp\left(-\frac{1}{2}\var[X]\right).
    \qedhere
\]
\end{proof}

\subsection{Fouvry primes}
For integer $x$, let ${\cal P}^+(x)$ be the largest prime dividing $x$. Say that prime $P$ is a
\emph{Fouvry prime} if ${\cal P}^+(P-1) > P^{2/3}$.
\begin{theorem}[Fouvry \cite{fouvry}]
    \label{thm:fouvry}
    There exist constants $c > 0$ and $n_0 > 0$ such that for all $x > n_0$,
    the number of Fouvry primes smaller than $x$ is at least
    $c \frac{x}{\log x}$.
\end{theorem}
Following \cite{cai,discrete_log}, we will prove our negative results for Shor's algorithms for Fouvry primes $P$, as the assumption that $P-1$ has a
large prime factor is very useful for our bit distribution arguments.
Since the number of primes at most $x$ is asymptotically $\frac{x}{\log x}$, 
Fouvry's theorem states
that the set of Fouvry primes $P$ has positive density in the set of all primes, and hence
the algorithms fail for a positive density of all primes.

\section{Factoring}
\label{sec:factoring}
\subsection{Shor's quantum factoring algorithm}
\label{sec:factoring_algo}
Shor's quantum factoring algorithm \cite{shor} solves the following problem: 
given $n$-bit integer $N$ promised
to be the product of unknown primes $P$ and $Q$, find $P$ and $Q$.
The algorithm begins by choosing a random $x \bmod{N}$ and preparing the state
\begin{equation}
    \frac{1}{2^{n/2}} \sum_{u=0}^{2^n-1} \ket{u} \ket{x^u \bmod{N}}
\end{equation}
in two $n$-qubit registers. The algorithm aims to find the order $\omega$ of $x$ in the multiplicative 
group $\mathbb{Z}_N^*$ of integers $\bmod{\,N}$ -- that is, $\omega$ is the smallest number in
$\mathbb{Z}_n \setminus \{0\}$ such that $x^\omega \equiv 1 \bmod{N}$.
From $\omega$, the algorithm recovers $P$ and $Q$ via classical postprocessing.
To find $\omega$, the algorithm applies an $n$-qubit QFT $F_{2^n}$ to the first
register. The resulting state is
\begin{align*}
    &\frac{1}{2^n} \sum_{u=0}^{2^n-1} \sum_{v=0}^{2^n-1} \exp\left(\frac{2\pi \ii uv}{2^n}\right)\ket{v}
    \ket{x^u \bmod{N}} \\
    &= \frac{1}{2^n} 
    \sum_{v=0}^{2^n-1} \ket{v} \left(\sum_{u^*=0}^{\omega-1} \left(
    \sum_{\substack{u \in [2^n] \\ u \equiv u^* \bmod{\omega}}} 
    \exp\left(\frac{2\pi \ii uv}{2^n}\right)\right)\ket{x^{u^*} \bmod{N}}\right).
\end{align*}
Finally, the algorithm measures the state in the first register. The probability of measuring
$\ket{v}$ is
\begin{align*}
    \frac{1}{2^{2n}} \sum_{u^*=0}^{\omega-1} 
    \left|\sum_{\substack{u \in [2^n] \\ u \equiv u^* \bmod{\omega}}} 
    \exp\left(\frac{2\pi \ii uv}{2^n}\right)\right|^2
    &= \frac{1}{2^{2n}} \sum_{u^*=0}^{\omega-1} 
    \left|\sum_{\substack{k \geq 0 \\ u^* + k\omega < 2^n}} 
    \exp\left(\frac{2\pi \ii (u^*+k\omega)v}{2^n}\right)\right|^2 \\
    &= \frac{\omega}{2^{2n}} 
    \left|\sum_{k=0}^{K-1} 
    \exp\left(\frac{2\pi \ii u_kv}{2^n}\right)\right|^2,
    \numberthis\label{eq:no_error_factoring}
\end{align*}
where $u_k := u^* + k\omega$ and $K \approx \frac{2^n}{\omega}$ is the largest integer such that $u^* + (K-1)\omega < 2^n$
(note that the norm of the sum in \eqref{eq:no_error_factoring} does not depend on $u^*$).

Let $\{z\}_{2^n}$ be the residue of $z \bmod 2^n$ in the range 
$-2^{n-1} < \{z\}_{2^n} \leq 2^{n-1}$. Shor considers states $\ket{v}$ satisfying
\begin{equation}
    |\{\omega v\}_{2^n}| \leq \frac{\omega}{2}.
    \label{eq:v_bound}
\end{equation}
If $v$ satisfies \eqref{eq:v_bound}, then, in the absence of noise, 
the terms in the sum in \eqref{eq:no_error_factoring} are 
concentrated in one half of the unit
circle, so the probability of measuring a $\ket{v}$ satisfying \eqref{eq:v_bound} is high. Shor
also shows that, with high probability, we recover a factor of $N$ from a $v$ satisfying 
\eqref{eq:v_bound}. In \autoref{thm:factor}, we analyze the probability of measuring some $v$
satisfying \eqref{eq:v_bound} when noise is present in the QFT used by Shor's factoring algorithm.
We show that, if the noise level is sufficiently small, the expectation over the random noise distribution of the probability 
of measuring some $v$ that satisfies \eqref{eq:v_bound} is at 
least $1/\poly(n)$, so the algorithm still runs in expected polynomial time. But if $P$ and
$Q$ are Fouvry primes and the noise
level asymptotically exceeds a certain threshold as a function of $n$, then 
the expectation of this probability is
superpolynomially small.

Note that $v \in [2^n]$ satisfies \eqref{eq:v_bound} if and only if there is a $j \in [\omega]$
such that $|\omega v - 2^nj| \leq \frac{\omega}{2}$, or equivalently
$|v - \frac{2^n}{\omega} j| \leq \frac{1}{2}$. In other words, $v$ satisfies \eqref{eq:v_bound} if
and only if $v = \lfloor \frac{2^n}{\omega} j \rceil$
(the closest integer to $\frac{2^n}{\omega} j$) for some $j \in [\omega]$.
\begin{definition} \label{def:pccj}
    For $j \in [\omega]$ and $c,c' > 0$, say that the predicate
    $\varphi_{c,c'}(j)$ is true if there is a set $T_j \subset \{b,\ldots,n\}$ such that
    \begin{equation}
        \pr_{(k,k') \sim [K] \times [K]} [|V_{j,k,k'}| \geq cn] \geq 1-2^{-c'n},
        \label{eq:prk}
    \end{equation}
    where $(k,k') \sim [K] \times [K]$ means that $(k,k')$
    is sampled uniformly at random from $[K] \times [K]$ and,
    with $v_j = \lfloor \frac{2^n}{\omega} j \rceil$,
    \[
        V_{j,k,k'} := \{i \in T_j \mid (u_k^{[i-b]} - u_{k'}^{[i-b]}) v_j^{[n-i]} = \pm 1\}.
    \]
\end{definition}
In other words, $\varphi_{c,c'}(j)$ is true if there is a high probability over $(k,k')$
that there is a linear (in $n$) size overlap in the 1 bits of the binary expansion of $v_j$
and the symmetric difference of the binary expansions of $u_k$ and $u_{k'}$. Every such
overlapping 1 bit will `activate' a random noise variable $r_0^{(n-i)}$ in the analysis of
the noisy version of \eqref{eq:no_error_factoring}.
The following lemma, proved in \cite{cai} via analysis of the bit distributions of
$u_k$ and $v_j$ for Fouvry primes $P$ and $Q$ 
shows that there exist $c,c' > 0$ such that
it is exponentially unlikely over the choice of $j \in [\omega]$ that $\varphi_{c,c'}$ does not
hold.
\begin{lemma}[{\cite[Lemma 4]{cai}}]
    \label{lem:prj}
    Suppose $N = PQ$ for (odd) Fouvry primes $P,Q$ and $b = O(\log n)$. Then
    there exist constants $c, c', c'' > 0$ such $\pr_{j \sim [\omega]}[\varphi_{c,c'}(j)] \geq 1-2^{-c''n}$, where the probability is for a uniformly sampled $j$ from $[\omega]$.
\end{lemma}
We now prove the main theorem of this section.
\begin{theorem}\label{thm:factor}
    Suppose each controlled-$R_{2^{-k}}$-gate in the quantum
    Fourier transform circuit is replaced by a controlled-$\widetilde{R_{2^{-k}}}$-gate for 
    all $k\geq b$, with $b = O(\log n)$ and noise magnitude parameter $\epsilon$. 
    Let $N = PQ$ have binary length $n$. 
    If 
    \[
        \frac{\epsilon}{2^b} = O\left(\left(\frac{\log n}{n}\right)^{1/2}\right),
    \]
    then, under this error model, 
    the expectation over the random noise of the probability of
    Shor's algorithm measuring a $\ket{v}$ satisfying \eqref{eq:v_bound} is at least $1/\poly(n)$.
    Hence the algorithm finds a prime factor of $N$ in expected time polynomial in $n$.

    On the other hand, if
    \[
        \lim_{n\to\infty} \left(\frac{\epsilon}{2^b}\right) \Big/ \left(\frac{\log n}{n}\right)^{1/2} = \infty,
    \]
    then there is a set $S$ of primes of positive density such that, for primes $P,Q \in S$, the expectation over
    the random noise of the probability of
    Shor's algorithm measuring a $\ket{v}$ satisfying \eqref{eq:v_bound}, which is used to 
    factor $N = PQ$, is exponentially small in $n$.
\end{theorem}
\begin{proof}
For $j \in [\omega]$, write $v_j = \lfloor \frac{2^n}{\omega} j \rceil = \frac{2^n}{\omega} j + \eta_j$ with $|\eta_j| \leq \frac{1}{2}$, and let $p_j$ be the probability that the measurement
in Shor's factoring algorithm yields $\ket{v_j}$.
Then, by the discussion before \autoref{def:pccj},
\begin{equation} \label{eq:p}
    p := \sum_{j \in [\omega]} p_j
\end{equation}
is the probability that the quantum measurement yields a state $\ket{v}$ satisfying
\eqref{eq:v_bound}, which we view as the algorithm's success probability. We aim to give
lower and upper bounds for $\ex[p]$ (the expectation of $p$ over the random noise), 
respectively, at the two noise levels in the theorem 
statement. Recall that, if Shor's algorithm is allowed to use the noise-free QFT $F_{2^n}$, 
the probability that the quantum measurement step
algorithm yields $\ket{v} = \ket{v_j}$ is given by \eqref{eq:no_error_factoring}. 
The first author \cite[Equation 4]{cai} showed via a calculation similar to that leading to 
\eqref{eq:prob_dl} below that, if the algorithm instead uses the noisy QFT
$\widetilde{F_{2^n}}$, then
\begin{align*}
    p_j = &\frac{\omega}{2^{2n}} \left|
        \sum_{k=0}^{K-1} \exp\left(
        2\pi \ii \left[
            \sum_{t=1}^n \frac{\sum_{s=0}^{n-t} u^{[n-t-s]}_k v_j^{[s]}}{2^t}
        + \frac{\epsilon}{2^b} E_k
        \right] \right) \right|^2 \\
    &= \frac{\omega}{2^{2n}} \left| 
        \sum_{k=0}^{K-1} \exp\left(
            2\pi \ii \left[ \frac{u_k v_j}{2^n}
        + \frac{\epsilon}{2^b} E_k
        \right] \right) \right|^2 \\
    &= \frac{\omega}{2^{2n}} \left|
        \sum_{k=0}^{K-1} \exp\left(
            2\pi \ii \left[ \frac{(u^*+k\omega) (\frac{2^n}{\omega}j + \eta_j)}{2^n}
        + \frac{\epsilon}{2^b} E_k
        \right] \right) \right|^2 \\
    &= \frac{\omega}{2^{2n}} \left| 
    \exp\left(2\pi \ii \frac{u^*(\frac{2^n}{\omega}j + \eta_j)}{2^n}\right)
        \sum_{k=0}^{K-1} \exp\left(
            2\pi \ii \left[ \frac{2^n kj + k\omega \eta_j}{2^n}
        + \frac{\epsilon}{2^b} E_k
        \right] \right) \right|^2 \\
    &= \frac{\omega}{2^{2n}} \left| 
        \sum_{k=0}^{K-1} \exp\left(
            2\pi \ii \left[\frac{A_j}{K} k
        + \frac{\epsilon}{2^b} E_k
        \right] \right) \right|^2
        \numberthis \label{eq:sumnorm}
\end{align*}
where $A_j = K\frac{\omega}{2^n} \eta_j$ and
\[
    E_k := \sum_{i=0}^{n-b} \left(\sum_{\ell=0}^{n-b-i} \frac{u_k^{[n-b-i-\ell]}r^{(i)}_\ell}{2^\ell}\right)v_j^{[i]}.
\]
The situation in \eqref{eq:sumnorm} is similar to \cite[Lemma 3]{jozsa_notes_2003}, but with error modeled by a normal
distribution instead of a bounded deterministic function. 
The sum is composed
of $K$ points evenly spaced between 0 and $2\pi \frac{A_j}{K}(K-1) \approx 2\pi A_j$ radians on the unit circle, 
with the $k$th term perturbed by independent noise $\frac{\epsilon}{2^b}E_k$.

Following \cite{cai}, we now expand the square norm expression in \eqref{eq:sumnorm} and use
the fact that $\exp(\mathfrak i\varphi) + \exp(-\mathfrak i\varphi) = 2\cos(\varphi)$ to obtain
\begin{align*}
    p_j &= \frac{\omega}{2^{2n}} \left(K + 2\sum_{0 \leq \ell < k < K}
        \cos\left( 2\pi \left[ \frac{A_j}{K}(k-\ell)
            + \frac{\epsilon}{2^b} (E_k-E_\ell)
        \right] \right) \right) \\
        &= \frac{K\omega}{2^{2n}} + \frac{2\omega}{2^{2n}}
        \sum_{\ell=0}^{K-2}
        \sum_{k=1}^{K-1-\ell}
        \cos\left( 2\pi \left[\frac{A_j}{K} k
            + \frac{\epsilon}{2^b} (E_{\ell+k}-E_\ell)
        \right] \right) \\
        &= \frac{K\omega}{2^{2n}} + \frac{2\omega}{2^{2n}}\sum_{\ell=0}^{K-2}
        \sum_{k=1}^{K-1-\ell}\cos\left(\frac{2\pi A_j}{K} k
                +  X_{k,\ell}
         \right),
        \numberthis\label{eq:double_sum}
\end{align*}
where $X_{k,\ell} := \frac{2\pi\epsilon}{2^b}(E_{\ell+k}-E_\ell)$. Here, $X_{k,\ell}$ is a
sum of independent Gaussian random variables, so is Gaussian.
Applying \autoref{fact:ex} in \eqref{eq:double_sum}, we find
\begin{align*}
    \ex[p_j] &= \frac{K\omega}{2^{2n}} + \frac{2\omega}{2^{2n}}
    \sum_{\ell=0}^{K-2} \sum_{k=1}^{K-1-\ell}
    \cos\left(\frac{2\pi A_j}{K} k \right)
    \exp\left(-\frac{1}{2}\var[X_{k,\ell}]\right).
    \numberthis\label{eq:pvlb}
\end{align*}

Towards bounding $\var[X_{k,\ell}]$,
we next analyze the variance of the difference of two noise terms $E_{\ell+k},E_{\ell}$.
For the positive claim in the theorem statement,
we provide an upper bound on the variance, and for the negative claim, we provide a lower 
bound that holds with high probability (due to \autoref{lem:prj}) 
over the choice of $0 \le \ell < \ell + k  \in [K]$. For the upper bound, choose any $k,k' \in [K]$.
Each $u_k^{[\ell]}, u_{k'}^{[\ell]}, v^{[i]}_j \in \{0,1\}$ and the random variables $r_\ell^{(i)} \sim N(0,1)$ are independent for distinct pairs $(i,\ell)$, so
\begin{align*}
    \var\left(\frac{2\pi\epsilon}{2^b}(E_k-E_{k'})\right) 
    &= \sum_{i=0}^{n-b} \sum_{\ell=0}^{n-b-i} 
    \var\left(\frac{2\pi\epsilon}{2^b}\frac{(u_k^{[n-b-i-\ell]}-u_{k'}^{[n-b-i-\ell]})v_j^{[i]}}{2^\ell}r^{(i)}_\ell\right) \\
    &\leq (n-b+1) \sum_{\ell=0}^{n-b} \left(\frac{2\pi\epsilon}{2^b}\right)^2 \frac{1}{4^\ell}\\
    &\leq 2(n-b+1) \left(\frac{2\pi\epsilon}{2^b}\right)^2.
    \numberthis \label{eq:variance}
\end{align*}
For the variance lower bound, we assume that there exist $c,c' > 0$ such that the predicate
$\varphi_{c,c'}(j)$ is satisfied (since $P$ and $Q$ are Fouvry primes, \autoref{lem:prj} guarantees
that this is the case with high probability over choice of $j$).
It follows that $|V_{j,k,k'}| \geq cn$ (as in \eqref{eq:prk}) with
high probability over the choice of $k,k'$. 
For $k,k' \in [K]$ satisfying $|V_{j,k,k'}| \geq cn$, we have
\begin{align*}
    \var\left(\frac{2\pi\epsilon}{2^b}(E_k-E_{k'})\right) 
    &= \sum_{i=0}^{n-b} \sum_{\ell=0}^{n-b-i} 
    \var\left(\frac{2\pi\epsilon}{2^b}\frac{(u_k^{[n-b-i-\ell]}-u_{k'}^{[n-b-i-\ell]})v_j^{[i]}}{2^\ell}r^{(i)}_\ell\right) \\
    &\geq \sum_{i=0}^{n-b}
    \var\left(\frac{2\pi\epsilon}{2^b}(u_k^{[n-b-i]}-u_{k'}^{[n-b-i]})v_j^{[i]}r^{(i)}_0\right) \\
    &= \sum_{i=b}^{n}
    \var\left(\frac{2\pi\epsilon}{2^b}(u_k^{[i-b]}-u_{k'}^{[i-b]})v_j^{[n-i]}r^{(n-i)}_0\right) \\
    &= \sum_{i \in V_{j,k,k'}}
    \var\left(\frac{2\pi\epsilon}{2^b}r^{(n-i)}_0\right) \\
    &= |V_{j,k,k'}| \left(\frac{2\pi\epsilon}{2^b}\right)^2 \\
    &\geq cn \left(\frac{2\pi\epsilon}{2^b}\right)^2.
    \numberthis\label{eq:varlb}
\end{align*}

Next, we use the variance upper bound \eqref{eq:variance} to prove the lower bound on
$p$ in the case when 
$\frac{\epsilon}{2^b} = O\left(\left(\frac{\log n}{n}\right)^{1/2}\right)$.
We wish to consider only those $j \in [\omega]$ for
which $\cos\left(\frac{2\pi A_j}{K} k \right) \geq 0$ for every $k = 1,\ldots,K-1-\ell$, as in this case we can
apply the variance upper bound \eqref{eq:variance} to every
term in the sum to obtain a lower bound on \eqref{eq:pvlb}. 
Recall that $A_j = K\frac{\omega}{2^n} \eta_j$, so 
$\frac{2\pi A_j}{K} k = 2\pi\eta_j  \frac{k \omega}{2^n}$, and that $K$ was defined so that
$u^* + (K-1)\omega < 2^n$. The latter implies that
$\frac{k \omega}{2^n} < 1$ for every $k \leq K-1$, so
\[
    |\eta_j| < \frac{1}{4} \implies \cos\left(\frac{2\pi A_j}{K} k \right) \geq 0
    \text{ for every } k = 1,\ldots,K-1-\ell.
\]
Later we will give a lower bound for $\ex[p] = \sum_{j \in [\omega]} \ex[p_j]$.
Each $\ex[p_j]$ is a sum which may have negative terms; however
it only has nonnegative terms if $|\eta_j| < \frac{1}{4}$. Since
every $p_j \geq 0$, we will still obtain
a lower bound on $p$ by dropping
those $p_j$ for which $|\eta_j| \geq \frac{1}{4}$. Then we argue that
$|\eta_j| < \frac{1}{4}$ for
a constant fraction  of $j$, thus dropping those $p_j$ causes no harm.

So suppose that $|\eta_j| < \frac{1}{4}$. Then, applying \eqref{eq:variance} in
\eqref{eq:pvlb}, we obtain
\begin{equation} \label{eq:epj_lb}
    \ex[p_j] \geq  \frac{K\omega}{2^{2n}}+ \frac{2\omega}{2^{2n}}
    \exp\left(-(n-b+1) \left(\frac{2\pi\epsilon}{2^b}\right)^2\right)
    \sum_{\ell=0}^{K-2} \sum_{k=1}^{K-1-\ell}
    \cos\left(\frac{2\pi A_j}{K}k \right).
\end{equation}
Define
\[
    S := \sum_{\ell=0}^{K-2} \sum_{k=1}^{K-1-\ell} \cos\left(\frac{2\pi A_j}{K}k \right)
    = \sum_{k=1}^{K} (K-k) \cos\left(\frac{2\pi A_j}{K}k \right).
\]
First, if $A_j = 0$, then clearly $S = \Theta(K^2)$. Otherwise,
we view $S$ as a right-Riemann sum and approximate it with an integral. In general, the error of such an
approximation for a continuously differentiable function $f(x)$ on an interval $[a,b]$ is
\begin{equation}
    \left|\sum_{k=a+1}^{b} f(k) - \int_{a}^b f(x) dx \right| \leq 
    \left[\max_{a \leq x \leq b}|f'(x)|\right] (b-a).
    \label{eq:riemann}
\end{equation}
Here, with $f(x) := (K-x)\cos\left(\frac{2\pi A_j}{K}x \right)$, $a := 0$, and $b := K$, we have
\begin{equation}
    \left|S - \int_{0}^K (K-x) \cos\left(\frac{2\pi A_j}{K}x \right) dx \right|
    \leq 
    O(K),
    \label{eq:int_approx}
\end{equation}
as $|A_j| = O(1)$ uniformly for all $j$.
Evaluating the integral, we find
\begin{equation}
    \frac{K^2}{2} \ge \int_{0}^K (K-x) \cos\left(\frac{2\pi A_j}{K}x \right) dx 
    \ge  \int_{0}^{K/2} \frac{K}{2} \cos\left(\frac{2\pi A_j}{K}x \right) dx 
    =\Omega(K^2),
\end{equation}
as $|\frac{2\pi A_j}{K}x|
\le \frac{\pi}{4} \frac{K \omega}{2^n}$ for $x \leq K/2$
and so $\cos\left(\frac{2\pi A_j}{K}x \right) = \Omega(1)$ is positive and bounded away from 0.
It follows that $S = \Theta(K^2)$. Therefore \eqref{eq:epj_lb}, along with the fact that
$K \approx \frac{2^n}{\omega}$, give
\begin{align*}
    \ex[p_j] &=
    \Omega\left(\frac{K\omega}{2^{2n}} + \frac{2\omega}{2^{2n}}
    \exp\left(-(n-b+1) \left(\frac{2\pi\epsilon}{2^b}\right)^2\right) K^2\right) \\
             &= \Omega\left(\frac{1}{2^n} + \frac{1}{\omega}
    \exp\left(-(n-b+1) \left(\frac{2\pi\epsilon}{2^b}\right)^2\right)\right). 
    \numberthis\label{eq:pvlb2}
\end{align*}
The conclusion \eqref{eq:pvlb2} holds for those $j \in [\omega]$ satisfying 
$|\eta_j| \leq \frac{1}{4}$.
Recall that $|\eta_j| \leq \frac{1}{4}$ if and only if there is an integer $s$ such that
$|\frac{2^n}{\omega} j - s | \leq \frac{1}{4}$.
Write $\frac{2^n}{\omega} = \frac{2^{n'}}{\omega'}$ for odd 
$\omega' = \frac{\omega}{2^{n-n'}}$. If $\omega' = 1$, then $\omega$ is a power of two and
$\eta_j = 0$ for every $j$, as desired.
Otherwise, $|j  \frac{2^n}{\omega} - s | \leq \frac{1}{4}$ is 
equivalent to $|j  2^{n'} - s \omega' | \leq \omega'/4$, which in turn is equivalent to
$\{j 2^{n'}\}_{\omega'} \in \{-\lfloor \frac{\omega'}{4} \rfloor,\ldots, \lfloor \frac{\omega'}{4}\rfloor\}$.
Since $2^{n'}$ and $\omega'$ are coprime, $\{j 2^{n'}\}_{\omega'}$ takes all values mod
$\omega'$ as $j$ ranges in $[\omega']$. Similarly, for every $t \in [2^{n-n'}]$,
$\{j 2^{n'}\}_{\omega'}$ takes all values mod
$\omega'$ as $j$ ranges in $t \omega',\ldots,(t+1)\omega'-1$.
Hence the proportion of $j \in [\omega] = 
\{0, ..., \omega'-1, \omega', ..., 2\omega' -1, ..., 2^{n-n'}\omega' -1\}$ for which
$\eta_j < \frac{1}{4}$ is close to $\frac{1}{2}$; certainly it is at least $\frac{1}{3}$,
the value when $\omega' = 3$. Thus, applying \eqref{eq:pvlb2},
\begin{align*}
    \ex[p] = \sum_{j \in [\omega]} \ex[p_j] \geq \sum_{\substack{j \in [\omega] \\ |\eta_j| < 1/4}} \ex[p_j]
    &= \Omega\left(\omega \left[\frac{1}{2^n} + \frac{1}{\omega}
    \exp\left(-(n-b+1) \left(\frac{2\pi\epsilon}{2^b}\right)^2\right)\right]\right) \\
    &= \Omega\left(\frac{\omega}{2^n} +
    \exp\left(-(n-b+1) \left(\frac{2\pi\epsilon}{2^b}\right)^2\right)\right). 
    \numberthis\label{eq:exp_lb}
\end{align*}
Finally, under the assumption that
$\frac{\epsilon}{2^b} = O\left(\left(\frac{\log n}{n}\right)^{1/2}\right)$,
there is a constant $\gamma > 0$ such that
\[
    \exp\left(-(n-b+1) \left(\frac{2\pi\epsilon}{2^b}\right)^2\right) \geq \exp(-\gamma\log n) = n^{-\gamma}.
\]
Substituting into \eqref{eq:exp_lb}, we find
\[
    \ex[p] = \Omega(n^{-\gamma}) = 1/\poly(n)
    \quad \text{if} \quad 
    \frac{\epsilon}{2^b} = O\left(\left(\frac{\log n}{n}\right)^{1/2}\right).
\]

Next we turn to the upper bound on $\ex[p]$.
Assume $P$ and $Q$ are Fouvry primes, so, by \autoref{thm:fouvry}, $P$ and $Q$ are drawn from a set
of primes of positive density.
We first apply the variance lower bound \eqref{eq:varlb} to obtain an
upper bound on $\ex[p_j]$ for those $j$ satisfying $\varphi_{c,c'}(j)$.
Let 
\[
    L := \{(k,\ell) \mid 0 \leq \ell \leq K-2, 1 \leq k \leq K-1-\ell\},
\]
and suppose that
$\varphi_{c,c'}(j)$ holds. Then, splitting the double summation in \eqref{eq:pvlb}, upper bounding 
every cosine factor by 1, and applying \eqref{eq:prk} and \eqref{eq:varlb},
\begin{align*}
    \ex[p_j] &\leq
    \frac{K\omega}{2^{2n}} + \frac{2\omega}{2^{2n}} \left(
    \sum_{\substack{(k,\ell) \in L\\|V_{j,\ell+k,\ell}| \geq cn}}
    \exp\left(-\frac{1}{2}\var[X_{k,\ell}]\right)
    + \sum_{\substack{(k,\ell) \in L\\|V_{j,\ell+k,\ell}| < cn}}
    \exp\left(-\frac{1}{2}\var[X_{k,\ell}]\right) \right) \\
    &\leq \frac{K\omega}{2^{2n}} + \frac{2\omega}{2^{2n}} \left(
    \sum_{\substack{(k,\ell) \in L\\|V_{j,\ell+k,\ell}| \geq cn}}
    \exp\left(-\frac{c}{2}\left(\frac{2\pi\epsilon}{2^b}\right)^2 n\right)
    + \big|\{(k,\ell) \in L: |V_{j,\ell+k,\ell}| < cn\}\big| \right) \\
    &\leq \frac{K\omega}{2^{2n}} + \frac{2\omega}{2^{2n}} \left(
    |L| \cdot \exp\left(-\frac{c}{2}\left(\frac{2\pi\epsilon}{2^b}\right)^2 n\right)
    + K^22^{-c'n} \right).
\end{align*}
Using that $|L| \leq K^2$ and that $K \approx \frac{2^n}{\omega}$ then gives
\begin{equation} \label{eq:pvub2}
    \ex[p_j] = O\left(\frac{1}{2^n} + \frac{1}{\omega} \left[
    \exp\left(-\frac{c}{2} \left(\frac{2\pi\epsilon}{2^b}\right)^2 n\right)
    + 2^{-c'n} \right]\right).
\end{equation}
Now, under the assumption that 
$\frac{\epsilon}{2^b} \Big/ \left(\frac{\log n}{n}\right)^{1/2} \to \infty$,
\[
    \exp\left(-\frac{c}{2} \left(\frac{2\pi\epsilon}{2^b}\right)^2 n\right)
    = o(1/\poly(n)).
\]
Applying this to \eqref{eq:pvub2}, we find that, for $j$ such that $\varphi_{c,c'}(j)$ holds,
\[
    \ex[p_j] = O\left(\frac{1}{2^n} + \frac{1}{\omega} 2^{-c'n}\right) +
    o\left(\frac{1}{\omega} \frac{1}{\poly(n)}\right), 
\]
Furthermore, by \eqref{eq:sumnorm}, we have
$\ex[p_j] \leq \frac{\omega}{2^{2n}} K^2 \approx \frac{1}{\omega}$ for any $j$.
By \autoref{lem:prj}, there is a $c''$ such that 
$\pr_{j \sim [\omega]}[\varphi_{c,c'}(j)] \geq 1-2^{-c''n}$. Hence
\begin{align*}
    \ex[p] = \sum_{j \in [\omega]} \ex[p_j]
    &= \sum_{\substack{j \in [\omega] \\ \varphi_{c,c'}(j)}} \ex[p_j]
    + \sum_{\substack{j \in [\omega] \\ \neg \varphi_{c,c'}(j)}} \ex[p_j] \\
    &\leq \omega\left[O\left(\frac{1}{2^n} + \frac{1}{\omega} 2^{-c'n}\right) +
    o\left(\frac{1}{\omega} \frac{1}{\poly(n)}\right)\right]
    + 2^{-c''n}\omega \left[\frac{1}{\omega}\right] \\
    &= O\left(\frac{1}{K} +  2^{-c'n} + 2^{-c''n}\right) +
    o\left(\frac{1}{\poly(n)}\right)\\
    & = o\left(\frac{1}{\poly(n)}\right),
    \quad \text{if} \quad 
    \frac{\epsilon}{2^b} \Big/ \left(\frac{\log n}{n}\right)^{1/2} \to \infty.
\end{align*}
This is superpolynomially small.
\end{proof}

\section{Discrete Log}
\label{sec:discretelog}
\subsection{Shor's quantum discrete log algorithm}
\label{sec:shorsalgo}
In this section, we give an overview of Shor's quantum \cite{shor}
algorithm for the discrete log problem, following
some notations in
\cite{discrete_log}.
Given prime $P$, base $g \in \zz_P^*$ 
(the multiplicative group of integers mod $P$, which is a cyclic group of order
$P-1$), and $y \in \zz_P^*$, the discrete log problem is to find $0 \leq d \leq P-2$ such that $g^d \equiv y \bmod{P}$.
Choose $n$ such that $2^{n-1} \leq P < 2^n$.
Shor's algorithm begins by preparing the state
\begin{equation}
    \frac{1}{P-1} \sum_{u=0}^{P-2} \sum_{k=0}^{P-2} \ket{u} \ket{k} \ket{g^u y^{-k} \bmod P}
    = 
    \frac{1}{P-1} \sum_{u=0}^{P-2} \sum_{k=0}^{P-2} \ket{u} \ket{k} \ket{g^{u-d k} \bmod P}
    \label{eq:initialstate}
\end{equation}
in three $n$-qubit registers. Then the algorithm applies $n$-qubit quantum Fourier transforms 
to the first and second registers. The state becomes
\begin{equation}
    \frac{1}{2^n(P-1)} \sum_{u,k=0}^{P-2} \sum_{v,w=0}^{2^n-1}
    \exp(2\pi \ii \frac{uv + kw}{2^n}) \ket{v}\ket{w} \ket{g^{u-dk} \bmod{P}}.
    \label{eq:afterfourier}
\end{equation}
Now we measure the three registers. Let $0 \leq u^* \leq P-2$.
For each $0 \leq k \leq P-2$,
there is a unique integer $u_k \in [P-1]$, 
satisfying $ u_k - dk \equiv u^* \bmod P-1$.
Then, the probability of obtaining $\ket{v}\ket{w}\ket{g^{u^*}}$
upon measuring the three registers is
\begin{equation}
    \frac{1}{2^{2n}(P-1)^2} \left|\sum_{k=0}^{P-2} \exp(2\pi \ii \frac{u_kv + kw}{2^n}) \right|^2.
    \label{eq:perfectprob}
\end{equation}
Shor \cite{shor} shows that
\begin{equation}
    u_k v + kw \equiv u^*v + kT_v + V_k \bmod{2^n},
    \label{eq:ukvkw}
\end{equation}
where
\[
    T_v = vd + w - \frac{d}{P-1} \{v(P-1)\}_{2^n}
    \quad \text{and} \quad
    V_k = \left(\frac{kd}{P-1} - \left\lfloor \frac{kd+ u^*}{P-1}\right\rfloor\right) \{v(P-1)\}_{2^n}
\]
(while $T_v$ depends on $w$ and $V_k$ depends on $v$, we suppress these dependencies in 
the subscript, as they will not be important). Therefore we may substitute 
$u^*v + kT_v + V_k$ in the RHS of
\eqref{eq:ukvkw} into \eqref{eq:perfectprob} without changing the value.
We can also remove the factor $\exp (2 \pi \ii \frac{u^*v}{2^n})$, of norm one, which is independent  of $k$.
Consider states with $\ket{v}\ket{w}$ in the first two registers with $v$ and $w$ satisfying
\begin{equation}
    \left|\left\{T_v\right\}_{2^n}\right| \leq \frac{1}{2}
    \label{eq:vw1}
\end{equation}
and
\begin{equation}
    \left|\left\{v(P-1)\right\}_{2^n}\right| \leq \frac{2^n}{12}.
    \label{eq:vw2}
\end{equation}
Define
\[
    G = \{(v,w) \mid 0 \leq v,w < 2^n, \text{$v$ and $w$ satisfy \eqref{eq:vw1} and \eqref{eq:vw2}}\}.
\]
Shor \cite{shor} shows that, in the noise-free case, the probability of measuring some $\ket{v}\ket{w}$ satisfying $(v,w) \in G$ is at least a positive constant, and from such a pair one can extract the discrete log value $d$.
In 
\autoref{thm:dl}, we analyze the probability 
of measuring some $(v,w) \in G$ with noise present in the QFT. 
We show that, 
as an expectation over randomness of the noise,
if the noise level is sufficiently small, the probability of measuring some   $(v,w) \in G$ is at least
$\frac{1}{\poly(n)}$, so Shor's algorithm still runs in expected polynomial time. But if $P$
is a Fouvry prime and the noise
level asymptotically exceeds a certain threshold as a function of $n$, then this probability is
superpolynomially small.

Let $\pi_1(G) = \{v \mid (v,w) \in G\}$ be the projection of $G$ onto the first coordinate.
For any $v$, there is exactly one integer $0 \leq w < 2^n$ satisfying \eqref{eq:vw1}
(hence the suppression of $w$ in the subscript of $T_v$), so
$|\pi_1(G)| = |G|$. Shor \cite{shor} showed that $|\pi_1(G)|$
is approximately $2^n/6$, and is always at least $2^n/12$. In summary,
\begin{equation}
    2^n/12 \leq |\pi_1(G)| = |G| \leq 2^n.
    \label{eq:size_g}
\end{equation}

We now give a definition analogous to \autoref{def:pccj} from the analysis for factoring.
\begin{definition}
\label{def:prk_dl}
For $v \in \pi_1(G)$ and $c,c' > 0$, say that the predicate $\varphi_{c,c'}(v)$ is true if
\begin{equation}
    \pr_{(k,k') \sim [P-1] \times [P-1]} \big[|J_{v,k,k'}| \geq cn\big] \geq 1-2^{-c'n},
    \label{eq:prk_dl}
\end{equation}
where, for $v \in \pi_1(G)$,
\[
    J_{v,k,k'} =
    \left\{i \in [n-b+1] ~\big|~ (u_k^{[n-b-i]} - u_{k'}^{[n-b-i]})v^{[i]} = \pm 1\right\}.
\]
\end{definition}
The following lemma is shown
over the course of the proof\footnote{
Specifically, what we call $J_{v,k,k'}$ here  is denoted by $J(k) \Delta J(k')$ in \cite{discrete_log}.
Using the notation of that paper, it is shown before \cite[Equation 15]{discrete_log} that
the proportion of pairs $(k,k')$ for which $|J_{k} \Delta J_{k'}| \geq \delta_2|S_v|$ 
is  $1-O(2^{-c'n})$ for some $c' > 0$,
where $\delta_2 = 1/64$ is a constant and $S_v \subset [n]$.
Furthermore,
for a large subset $G'$ of $G$, 
if $v \in\pi_1(G')$, then $\delta_2 |S_v| = \Omega(n)$, and it is shown in \cite[Equation 11]{discrete_log}
that there is a constant $c'' > 0$ such that $\frac{|\pi_1(G')|}{|\pi_1(G)|} \geq 1-2^{c''n}$.
}
of \cite[Theorem 3.2]{discrete_log}.
\begin{lemma} \label{lem:prk_dl}
Suppose $P$ is an (odd) Fouvry prime and $b = O(\log n)$.
Then, for all but an exponentially small fraction of inputs
$y \in \zz_P^*$, there exist constants $c,c',c'' > 0$ such that $\pr_{v \sim \pi_1(G)}[\varphi_{c,c'}(v)] \geq 1-2^{-c''n}$,
where the probability is for a uniformly sampled $v$ from $\pi_1(G)$.
\end{lemma}
Now we prove the main theorem of this section, a discrete log analogue of \autoref{thm:factor}.
\begin{theorem}\label{thm:dl}
    Suppose each controlled-$R_{2^{-k}}$-gate in the quantum
    Fourier transform circuit is replaced by a controlled-$\widetilde{R_{2^{-k}}}$-gate for all $k\geq b$,
    with $b = O(\log n)$ and noise magnitude parameter $\epsilon$. 
    Let $P$ be a prime with binary length $n$. If
    \[
        \frac{\epsilon}{2^b} = O\left(\left(\frac{\log n}{n}\right)^{1/2}\right),
    \]
    then, under this error model, for any generator $g$ and any $y \in \zz_P^*$,
    the expectation over the random noise of the probability of
    Shor's algorithm measuring some $(v,w) \in G$ is at least 1/poly$(n)$. Hence the algorithm finds the discrete log value $d$ satisfying
    $g^d = y \bmod{P}$ in expected  polynomial time in $n$.

    On the other hand, if
    \[
        \lim_{n\to\infty} \left(\frac{\epsilon}{2^b}\right) \Big/ \left(\frac{\log n}{n}\right)^{1/2} = \infty,
    \]
    then, for a positive density of primes $P$, for any generator $g \in \zz_P^*$, and all but an 
    exponentially small fraction of inputs
    $y \in \zz_P^*$, the expectation over the random noise of the 
    probability of
    Shor's algorithm  measuring some $(v,w) \in G$, which is used to find the discrete log value $d$ 
    satisfying $g^d = y \bmod{P}$,
    is  exponentially small in $n$.
\end{theorem}
\begin{proof}
Henceforth, we assume $(v,w) \in G$.
Define
\begin{align*}
    p = \sum_{\substack{(v,w) \in G \\ u^* \in [P-1]}} p(v,w,g^{u^*})
    \numberthis\label{eq:define-expected-prob-p}
\end{align*}
to be probability that the measured quantum state
$\ket{v}\ket{w}\ket{u^*}$ satisfies $(v,w) \in G$, which we view as the algorithm's success 
probability. We aim to give
lower and upper bounds for $\ex[p]$ (the expectation of $p$ over the random noise), 
respectively, at the two noise levels in the theorem 
statement. In the noise-free case, we applied $F_{2^n}$ to the first two registers of the state in \eqref{eq:initialstate},
producing the state in \eqref{eq:afterfourier}. Here, we instead apply
$\widetilde{F_{2^n}}$ to both registers. Call the independent random variables associated with these two noisy 
QFTs $r^{(\cdot)}_\cdot$ and $\rho^{(\cdot)}_\cdot$, respectively.
We then measure the three registers. It is shown in \cite[Equation 8]{discrete_log} that
the probability of measuring $\ket{v}\ket{w}\ket{g^{u^*}}$ after the noisy transform is
\begin{align*}
    & p(v,w,g^{u^*}) \\
    &= \frac{1}{2^{2n} (P-1)^2} \Bigg| \sum_{k=0}^{P-2} \exp \Bigg(
        2 \pi \ii \Bigg[
    \sum_{t=0}^{n-1} v^{[t]} \left(0.u^{[n-t-1]}_k\ldots u_k^{[0]}\right)
        + \sum_{\tau=0}^{n-1} w^{[\tau]} \left(0.k^{[n-\tau-1]} \ldots k^{[0]}\right)
    + \frac{\epsilon}{2^b} E_k \Bigg] \Bigg) \Bigg|^2 \\
    &= \frac{1}{2^{2n} (P-1)^2} \Bigg| \sum_{k=0}^{P-2} \exp \Bigg(
        2 \pi \ii \Bigg[ \sum_{t=0}^{n-1} v^{[t]} u_k 2^{t-n}
        + \sum_{\tau=0}^{n-1} w^{[\tau]} k 2^{\tau-n}
    + \frac{\epsilon}{2^b} E_k \Bigg] \Bigg) \Bigg|^2 \\
    &= \frac{1}{2^{2n} (P-1)^2} \Bigg| \sum_{k=0}^{P-2} \exp \Bigg(
        2 \pi \ii \Bigg[ \frac{1}{2^n} u_k\sum_{t=0}^{n-1} v^{[t]} 2^{t}
        + \frac{1}{2^n} k \sum_{\tau=0}^{n-1} w^{[\tau]} 2^{\tau}
    + \frac{\epsilon}{2^b} E_k \Bigg] \Bigg) \Bigg|^2 \\
    &= \frac{1}{2^{2n} (P-1)^2} \Bigg| \sum_{k=0}^{P-2} \exp \Bigg(
        2 \pi \ii \Bigg[ \frac{u_k v + kw}{2^n} + \frac{\epsilon}{2^b} E_k \Bigg] \Bigg) \Bigg|^2,
        \numberthis\label{eq:prob_dl}
\end{align*}
with error terms
\[
    E_k := \sum_{i=0}^{n-b} \left(\sum_{\ell=0}^{n-b-i} \frac{u_k^{[n-b-i-\ell]} r^{(i)}_\ell}{2^\ell}\right)v^{[i]}
    + \sum_{i=0}^{n-b} \left(\sum_{\ell=0}^{n-b-i} \frac{k^{[n-b-i-\ell]} \rho^{(i)}_\ell}{2^\ell}\right)w^{[i]}.
\]
Now, substituting $u^*v + kT_v + V_k$ from \eqref{eq:ukvkw} for
$u_k v + kw$ in \eqref{eq:prob_dl}, then defining 
$A_v := \frac{P-1}{2^n} \{T_v\}_{2^n}$
and $B_k := \frac{V_k}{2^n}$, we have
\begin{align*}
    p(v,w,g^{u^*})
    &= \frac{1}{2^{2n} (P-1)^2} \Bigg| \sum_{k=0}^{P-2} \exp \Bigg(
        2 \pi \ii \Bigg[ \frac{u^*v + kT_v + V_k}{2^n} + \frac{\epsilon}{2^b} E_k \Bigg] \Bigg) \Bigg|^2 \\
    &= \frac{1}{2^{2n} (P-1)^2} \Bigg| \sum_{k=0}^{P-2} \exp \Bigg(
    2 \pi \ii \Bigg[ \frac{A_v}{P-1}k + B_k + \frac{\epsilon}{2^b} E_k \Bigg] \Bigg) \Bigg|^2.
    \numberthis\label{eq:sumnorm_dl}
\end{align*}
The situation in \eqref{eq:sumnorm_dl} is analogous to \eqref{eq:sumnorm}. 
Expanding the square norm expression in \eqref{eq:sumnorm_dl} gives
\begin{align*}
    &p(v,w,g^{u^*})  \\
    &= \frac{1}{2^{2n}(P-1)^2} \left(P-1 + 2\sum_{0 \leq \ell < k < P-1}
        \cos\left( 2\pi \left[ \frac{A_v}{P-1}(k-\ell)
            + B_k - B_\ell + \frac{\epsilon}{2^b} (E_k-E_\ell)
        \right] \right) \right) \\
       &= \frac{1}{2^{2n}(P-1)} + \frac{2}{2^{2n}(P-1)^2}
        \sum_{\ell=0}^{P-3}
        \sum_{k=1}^{P-2-\ell}
        \cos\left( 2\pi \left[\frac{A_v}{P-1} k
        + B_{\ell+k} - B_\ell + \frac{\epsilon}{2^b} (E_{\ell+k}-E_\ell)
        \right] \right) \\
       &= \frac{1}{2^{2n}(P-1)} + \frac{2}{2^{2n}(P-1)^2}
        \sum_{\ell=0}^{P-3}
        \sum_{k=1}^{P-2-\ell}
        \cos\left( 2\pi \left[\frac{A_v}{P-1} k
        + B_{\ell+k} - B_\ell \right] + X_{k,\ell}
        \right)
        \numberthis\label{eq:double_sum_dl}
\end{align*}
where 
\[
    X_{k,\ell} := \frac{2\pi\epsilon}{2^b} (E_{\ell+k}-E_\ell).
\]
Applying \autoref{fact:ex} in \eqref{eq:double_sum_dl}, we find
\begin{align*}
    &\ex[p(v,w,g^{u^*})] \\
    &= \frac{1}{2^{2n}(P-1)} + \frac{2}{2^{2n}(P-1)^2}
    \sum_{\ell=0}^{P-3} \sum_{k=1}^{P-2-\ell}
    \cos\left(2\pi \left[\frac{A_v}{P-1} k + B_{\ell+k} - B_\ell \right]\right)
    \exp\left(-\frac{1}{2}\var[X_{k,\ell}]\right).
    \numberthis\label{eq:pvlb_dl}
\end{align*}
Now we bound $\var[X_{k,\ell}]$.
For any $k,k' \in [P-1]$, since
each $u_k^{[\ell]}, u_{k'}^{[\ell]}, k^{[\ell]}, k'^{[\ell]}, v^{[i]}_j, w^{[i]}_j \in \{0,1\}$ and $r_\ell^{(i)}, \rho_\ell^{(i)} 
\sim N(0,1)$,
\begin{align*}
    &\var\left(\frac{2\pi\epsilon}{2^b}(E_k-E_{k'})\right) \\
    &= \sum_{i=0}^{n-b} \sum_{\ell=0}^{n-b-i} \left[
    \var\left(\frac{2\pi\epsilon}{2^b}\frac{(u_k^{[n-b-i-\ell]}-u_{k'}^{[n-b-i-\ell]})v_j^{[i]}}{2^\ell}r^{(i)}_\ell\right)
+ \var\left(\frac{2\pi\epsilon}{2^b}\frac{(k^{[n-b-i-\ell]}-k'^{[n-b-i-\ell]})w_j^{[i]}}{2^\ell}\rho^{(i)}_\ell\right) \right]\\
    &\leq 2(n-b+1) \sum_{\ell=0}^{n-b-i} \left(\frac{2\pi\epsilon}{2^b}\right)^2 \frac{1}{4^\ell}\\
    &\leq 4(n-b+1) \left(\frac{2\pi\epsilon}{2^b}\right)^2.
    \numberthis \label{eq:variance_dl}
\end{align*}
On the other hand, if $v,k,k'$ satisfy $|J_{v,k,k'}| \geq cn$ (as in \eqref{eq:prk_dl}), then
\begin{align*}
    &\var\left(\frac{2\pi\epsilon}{2^b}(E_k-E_{k'})\right) \\
    &\geq \sum_{i=0}^{n-b} \left[
    \var\left(\frac{2\pi\epsilon}{2^b}(u_k^{[n-b-i]}-u_{k'}^{[n-b-i]})v_j^{[i]}r^{(i)}_0\right)
+ \var\left(\frac{2\pi\epsilon}{2^b}(k^{[n-b-i]}-k'^{[n-b-i]})w_j^{[i]}\rho^{(i)}_0\right) \right]\\
&\geq \sum_{i \in J_{v,k,k'}} 
    \var\left(\frac{2\pi\epsilon}{2^b}r^{(i)}_0\right)\\
    &= \big|J_{v,k,k'}\big|
    \left(\frac{2\pi\epsilon}{2^b}\right)^2 \\
    &\geq  cn \left(\frac{2\pi\epsilon}{2^b}\right)^2.
    \numberthis\label{eq:varlb_dl}
\end{align*}

First, we use the variance upper bound \eqref{eq:variance_dl} to prove the lower bound on the 
quantity $p$ from \eqref{eq:define-expected-prob-p} in the case when 
$\frac{\epsilon}{2^b} = O\left(\left(\frac{\log n}{n}\right)^{1/2}\right)$.
As in the factoring analysis, we will consider only those pairs $(v,w) \in G$ for which
the argument of every cosine in \eqref{eq:pvlb_dl} is nonnegative.
Towards this, we will show that a constant fraction of $v \in \pi_1(G)$ 
(equivalently, by the discussion before \eqref{eq:size_g}, a constant fraction of
$(v,w) \in G$) satisfy
\begin{equation} \label{eq:triangle}
    \left|A_v \right| \leq \frac{1}{6} \text{ ~~~and~~~ }
    \left| B_{\ell+k} - B_\ell\right| \leq \frac{1}{12}
\end{equation}
for every $\ell,k$. Together, the two equations in \eqref{eq:triangle} imply that every 
cosine in \eqref{eq:pvlb_dl} is nonnegative as, for $k \leq P-2$, we have 
\[
    2\pi \left|\frac{A_v}{P-1} k + B_{\ell+k} - B_\ell\right| \leq 
    2\pi\left(\left|\frac{A_v}{P-1} k\right| + \left|B_{\ell+k} - B_\ell\right|\right)
    \leq 2 \pi \left(\frac{1}{6} + \frac{1}{12}\right) = \frac{\pi}{2}.
\]
First, regardless of $\ell$ and $k$,
\begin{align*}
    |B_{\ell+k} - B_{\ell}| 
    &= \frac{1}{2^n} |V_{\ell+k} - V_{\ell}| \\
    &= \frac{1}{2^n}\big|\{v(P-1)\}_{2^n}\big| \cdot \left|\frac{(\ell+k)d}{P-1} 
    - \left\lfloor \frac{(\ell+k)d+ u^*}{P-1}\right\rfloor - \frac{\ell d}{P-1} 
    + \left\lfloor \frac{\ell d+ u^*}{P-1}\right\rfloor \right| \\
    &\leq \frac{1}{12} \left|\frac{kd}{P-1} 
    + \left\lfloor \frac{\ell d+ u^*}{P-1}\right\rfloor
    - \left\lfloor \frac{kd + \ell d+ u^*}{P-1}\right\rfloor\right| \\
    & =: \frac{1}{12} |X + \lfloor Y \rfloor - \lfloor X+Y\rfloor|.
\end{align*}
We have 
\[
    -1 = X + (Y-1) - (X+Y) \leq X + \lfloor Y \rfloor - \lfloor X+Y\rfloor 
    \leq X + Y - \lfloor X+Y\rfloor \leq 1,
\]
so $|X + \lfloor Y \rfloor - \lfloor X+Y\rfloor| \leq 1$, and hence 
$|B_{\ell+k} - B_{\ell}| \leq \frac{1}{12}$.

Now to show \eqref{eq:triangle} it suffices to show that $|A_v| \leq \frac{1}{6}$
for a constant fraction of $v \in \pi_1(G)$.
Recall that $A_v = \frac{P-1}{2^n} \{T_v\}_{2^n}$, so $|A_v| \leq |\{T_v\}_{2^n}|$.
Since $(v,w) \in G$, we have $|\{T_v\}_{2^n}| \leq \frac{1}{2}$. Therefore
\begin{equation} \label{eq:a_over_p_bound}
    \left|A_v\right| \leq |\{T_v\}_1|
    = \left|\left\{vd + w + \frac{d}{P-1} \{v(P-1)\}_{2^n}\right\}_{1}\right|
    = \left|\left\{\frac{d}{P-1} \{v(P-1)\}_{2^n}\right\}_{1}\right|.
\end{equation}
Write $P-1 = a 2^m$ for odd $a$.
As $v$ ranges in $[2^n]$, $v(P-1) \bmod 2^n$ takes each value 
\[
    0,1\cdot 2^m,2\cdot 2^m, 3\cdot 2^m, \ldots, (2^{n-m} - 1) 2^m = 2^n - 2^m
\]
exactly $2^m$ times. When we iterate over only $v \in \pi_1(G)$,
which by \eqref{eq:vw2} (and the fact that every $v \in [2^n]$ satisfies \eqref{eq:vw1}
for some $w$) is equivalent to $|\{v(P-1)\}_{2^n}| \leq \frac{2^n}{12}$,
$\{v(P-1)\}_{2^n}$ takes every value in $\{j \cdot 2^m \mid j \in [-a',a']\}$ 
exactly $2^m$ times, where $a' = \lfloor \frac{2^{n-m}}{12} \rfloor$. By this fact and
\eqref{eq:a_over_p_bound}, to obtain the desired constant fraction of $v \in \pi_1(G)$ 
satisfying $|A_v| \leq \frac{1}{6}$, it suffices to show that a constant fraction 
of $j \in [-a',a']$ satisfy $|\{\frac{d}{P-1} j \cdot 2^m\}_1| \leq \frac{1}{6}$. Since
$P-1 = a2^m$ by definition, and the condition is symmetric about 0, this is equivalent to
\begin{equation} \label{eq:dja}
    |\{dj\}_a| \leq \frac{a}{6}
\end{equation}
for a constant fraction of $j \in [0,a']$.
We may reduce $d \bmod{a}$ to
assume that $d < a$ and further, by symmetry about 0, assume that $d \leq \frac{a}{2}$.

First suppose that $d < \frac{a}{3}$. 
Divide the full range of values mod $a$ into the following three intervals (of integer points)
of  approximate size $\frac{a}{3}$ (recall that $a$ is odd):
$\left[-\lfloor \frac{a}{2}\rfloor,-\lceil\frac{a}{6}\rceil\right]$,
$\left[-\lfloor \frac{a}{6}\rfloor,\lfloor\frac{a}{6}\rfloor\right]$,
$\left[\lceil \frac{a}{6}\rceil,\lfloor\frac{a}{2}\rfloor\right]$.
The three intervals are of size within
1 of each other, and the desired middle interval 
$\left[-\lfloor \frac{a}{6}\rfloor,\lfloor\frac{a}{6}\rfloor\right]$ has
 $2 \lfloor\frac{a}{6}\rfloor + 1 \ge \lfloor\frac{a}{3}\rfloor \ge d$
many integers, 
so the sequence $0,d,2d,\ldots,a'd$
visits the middle interval at least once on every cycle mod $a$ 
(separating the cycles by when the sequence passes $0 \bmod{a}$, which is contained in
the middle interval)
and visits the other two intervals at most one more time than -- hence at most twice as many times than -- the middle interval
on each cycle. 
Furthermore, the sequence starts at 0 in the middle interval then progresses over the
right half of the middle interval, so, even if the sequence ends before it returns to
the middle interval (that is, if $a'd < a-\lfloor\frac{a}{6}\rfloor$), 
it still includes at least half of the points in
the middle interval that it would have included had we extended the sequence until it
completed a cycle.
Therefore the fraction of $j \in [0,a']$ satisfying \eqref{eq:dja} is at least 
$\frac{1}{3} \cdot \frac{1}{2} \cdot \frac{1}{2} = \frac{1}{12}$.

Otherwise, $d \geq \frac{a}{3}$. Since $d \leq \frac{a}{2}$ as well, there is some
$\epsilon$ with $|\epsilon| \leq \frac{1}{2} \cdot \left(\frac{a}{2} - \frac{a}{3}\right) = \frac{a}{12}$ such that either $d = \frac{a}{3} + \epsilon$
or $d = \frac{a}{2} + \epsilon$. Now, considering every third term or every second term
of the sequence $0,d,2d,\ldots,a'd \bmod{a}$ gives a subsequence equivalent to
$0,3\epsilon,2(3\epsilon),\ldots \bmod{a}$ or
$0,2\epsilon,2(2\epsilon),\ldots \bmod{a}$, respectively (note that $3\epsilon$ and 
$2\epsilon$ are indeed integers in their respective cases).
Since $3|\epsilon|$ and $2|\epsilon|$ in respective cases are $ \leq 3\cdot \frac{a}{12} < \frac{a}{3}$, the reasoning in the case when
$d < \frac{a}{3}$  applies with $3\epsilon$ or $2\epsilon$ in place of $d$ (our target interval is symmetric about 0, and so we may replace
$\epsilon$ by $|\epsilon|$ in this analysis): at least $\frac{1}{12}$ of
the terms in the subsequence satisfy \eqref{eq:dja}, so the proportion of $j \in [0,a']$
satisfying \eqref{eq:dja} is at least $\frac{1}{3} \cdot \frac{1}{12} = \frac{1}{36}$.
Thus the proportion of $(v,w) \in G$ satisfying \eqref{eq:triangle} for every $\ell,k$
is at least $\frac{1}{36}$.

So suppose that $(v,w) \in G$ satisfy \eqref{eq:triangle}. Then every cosine in 
\eqref{eq:pvlb_dl} is nonnegative. Then applying the variance upper bound 
\eqref{eq:variance_dl} in \eqref{eq:pvlb_dl} gives
\begin{align*}
    &\ex[p(v,w,g^{u^*})] \\
    &\geq \frac{1}{2^{2n}P} + \frac{2}{2^{2n}P^2}
    \exp\left(-(n-b+1) \left(\frac{2 \pi \epsilon}{2^b}\right)^2\right)
    \sum_{\ell=0}^{P-3} \sum_{k=1}^{P-2-\ell}
    \cos\left(2\pi \left[\frac{A_v}{P-1} k + B_{\ell+k} - B_\ell \right]\right).
    \numberthis\label{eq:pvlb_dl2}
\end{align*}
Assume that $A_v \geq 0$; the proof for $A_v \leq 0$ is symmetric. 
By \eqref{eq:triangle}, each 
\[
    \cos\left(2\pi \left[\frac{A_v}{P-1} k + B_{\ell+k} - B_\ell \right]\right)
    \geq \cos\left(2\pi \left[\frac{A_v}{P-1} k + \frac{1}{12}\right]\right).
\]
Thus the double sum in \eqref{eq:pvlb_dl2} is lower-bounded by 
\[
    S := \sum_{\ell=0}^{P-3} \sum_{k=1}^{P-2-\ell}
    \cos\left(2\pi \left[\frac{A_v}{P-1} k + \frac{1}{12}\right]\right)
    = \sum_{k=1}^{P-1} (P-1-k) 
    \cos\left(2\pi \left[\frac{A_v}{P-1} k + \frac{1}{12}\right]\right).
\]
Now the situation is similar to the factoring analysis, and we apply an analogous integral
approximation to $S$, with $P-1$ in place of $K$. 
As in \eqref{eq:int_approx}, the error of the integral approximation is
$O(P)$. Evaluating the integral, we find
\begin{align*}
    \frac{(P-1)^2}{2} &\geq
    \int_{0}^{P-1} (P-1-x) 
    \cos\left(2\pi \left[\frac{A_v}{P-1} x + \frac{1}{12}\right]\right) dx \\
    &\geq \int_{0}^{(P-1)/2} \frac{P-1}{2} 
    \cos\left(\frac{\pi}{3}\right) dx \\
    &= \Omega(P^2),
\end{align*}
as, by \eqref{eq:triangle}, 
\[
    2\pi \left[\frac{A_v}{P-1} x + \frac{1}{12}\right] \leq 
    2\pi \left[\frac{1}{6} \cdot \frac{1}{2} + \frac{1}{12}\right] = \frac{\pi}{3}
\]
for $x \leq \frac{P-1}{2}$, so every cosine term is at least $1/2$.
Therefore $S = \Omega(P^2)$, so \eqref{eq:pvlb_dl2} gives
\begin{align*}
    \ex[p(v,w,g^{u^*})] 
    &= \Omega\left(\frac{1}{2^{2n}P} + \frac{2}{2^{2n}P^2}
    \exp\left(-2(n-b+1) \left(\frac{2\pi\epsilon}{2^b}\right)^2\right) P^2\right) \\
    &= \Omega\left(\frac{1}{2^{2n}}\left[\frac{1}{P} + 2
    \exp\left(-2(n-b+1) \left(\frac{2\pi\epsilon}{2^b}\right)^2\right)\right]\right)
    \numberthis\label{eq:pvlb_dl3}
\end{align*}
for every $(v,w) \in G$ satisfying \eqref{eq:triangle}.
Under the low noise assumption 
$\frac{\epsilon}{2^b} = O\left(\left(\frac{\log n}{n}\right)^{1/2}\right)$,
there is a constant $\gamma > 0$ such that
\[
    \exp\left(-2(n-b+1) \left(\frac{2\pi\epsilon}{2^b}\right)^2\right) \geq \exp(-\gamma\log n) = n^{-\gamma}.
\]
Recall that the number of $(v,w) \in G$ satisfying \eqref{eq:triangle} is at least
$\frac{1}{36} |G|$, which, by \eqref{eq:size_g}, is $\Omega(2^n)$.
Thus by \eqref{eq:pvlb_dl3} and the assumption $2^{n-1} \leq P < 2^n$,
\begin{align*}
    \ex[p] &= \sum_{\substack{(v,w) \in G \\ u^* \in [P-1]}} \ex[p(v,w,g^{u^*})]  \\
      &\geq \sum_{\substack{(v,w) \in G \text{ satisfy } \eqref{eq:triangle} \\ u^* \in [P-1]}} \ex[p(v,w,g^{u^*})] \\
      &= \Omega\left(2^n(P-1) \cdot \frac{1}{2^{2n}}\left[\frac{1}{P} + 2n^{-\gamma}\right]\right) \\
      &= \Omega(n^{-\gamma}),
    \text{ ~~~if~~~ } \frac{\epsilon}{2^b} = O\left(\left(\frac{\log n}{n}\right)^{1/2}\right).
\end{align*}
Thus, the expected success probability is lower bounded by
a fixed inverse polynomial.


Next, we prove the upper bound on $\ex[p]$ in the high noise case. 
Assume $P$ is a Fouvry prime; by \autoref{thm:fouvry}, the
set of all such primes has positive density. Let
\[
    L := \{(k,\ell) \mid 0 \leq \ell \leq P-3, 1 \leq k \leq P-2-\ell\},
\]
and suppose that
$\varphi_{c,c'}(v)$ holds. For ease of notation, let $D_{k,\ell} :=
2\pi \left[\frac{A_v}{P-1} k + B_{\ell+k} - B_\ell \right]$ denote
the cosine argument in \eqref{eq:pvlb_dl}.
Then, splitting the double summation in \eqref{eq:pvlb_dl}, upper bounding the absolute value of every
cosine factor by 1, and applying \eqref{eq:prk_dl} and \eqref{eq:varlb_dl},
\begin{align*}
    &\ex[p(v,w,g^{u^*})] \\
    &\leq \frac{1}{2^{2n}(P-1)} + \frac{2}{2^{2n}(P-1)^2} \Bigg(
    \sum_{\substack{(k,\ell) \in L\\|J_{v,k,k'}| \geq cn}}
    \exp\left(-\frac{1}{2}\var[X_{k,\ell}]\right)
    + \sum_{\substack{(k,\ell) \in L\\ |J_{v,ell+k, \ell}| < cn}}
    \exp\left(-\frac{1}{2}\var[X_{k,\ell}]\right) \Bigg) \\
    &\leq \frac{1}{2^{2n}(P-1)} + \frac{2}{2^{2n}(P-1)^2} \left(
    |L| \cdot \exp\left(-\frac{c}{2}\left(\frac{2\pi\epsilon}{2^b}\right)^2 n\right)
    + P^22^{-c'n} 
\right).
\end{align*}
Using the fact that $|L| < P^2$ then gives
\begin{equation}
    \ex[p(v,w,g^{u^*})] 
    = O\left(\frac{1}{2^{2n}}\left[\frac{1}{P} + \exp\left(-\frac{c}{2}\left(\frac{2\pi\epsilon}{2^b}\right)^2 n\right)\right] 
    + 2^{-c'n}\right)
    \label{eq:pvub2_dl}
\end{equation}
if $\varphi_{c,c'}(v)$ holds.
Now, under the assumption that 
$\frac{\epsilon}{2^b} \Big/ \left(\frac{\log n}{n}\right)^{1/2} \to \infty$, we have
\[
    \exp\left(-\frac{c}{2} \left(\frac{2\pi\epsilon}{2^b}\right)^2 n\right)
    = o(1/\poly(n)).
\]
Substituting into \eqref{eq:pvub2_dl}, if $v \in \pi_1(G)$ satisfies $\varphi_{c,c'}(v)$, then
\[
    \ex[p(v,w,g^{u^*})]
    = O\left(\frac{1}{2^{2n}}\left(\frac{1}{P} + 2^{-c'n}\right)\right) 
    + o\left(\frac{1}{2^{2n}}\frac{1}{\poly(n)}\right).
\]
Furthermore, \eqref{eq:sumnorm_dl} gives
$\ex[p(v,w,g^{u^*})] \leq \frac{1}{2^{2n}(P-1)^2} (P-1)^2 = \frac{1}{2^{2n}}$ for any $v,w,u^*$,
and, by \eqref{eq:size_g}, $|\pi_1(G)| = |G| \leq 2^n$.
By \autoref{lem:prk_dl}, there is a $c''$ such that 
$\pr_{v \sim \pi_1(G)}[\varphi_{c,c'}(v)] \geq 1-2^{-c''n}$. Therefore
\begin{align*}
    \ex[p] &= \sum_{\substack{(v,w) \in G:~ \varphi_{c,c'}(v) \\ u^* \in [P-1]}} \ex[p(v,w,g^{u^*})]
    + \sum_{\substack{(v,w) \in G:~ \neg \varphi_{c,c'}(v) \\ u^* \in [P-1]}} \ex[p(v,w,g^{u^*})] \\
      & \leq |G|(P-1)\left[O\left(\frac{1}{2^{2n}}\left(\frac{1}{P} + 2^{-c'n}\right)\right) 
    + o\left(\frac{1}{2^{2n}}\frac{1}{\poly(n)}\right)\right]
    + 2^{-c''n}|G|(P-1) \frac{1}{2^{2n}} \\
      &= O\left(\frac{1}{2^n} + 2^{-c'n} + 2^{-c''n}\right) + o\left(\frac{1}{\poly(n)}\right)\\
      &= o\left(\frac{1}{\poly(n)}\right),
      \text{ ~~~if~~~ } \frac{\epsilon}{2^b} \Big/ \left(\frac{\log n}{n}\right)^{1/2} \to \infty.
\end{align*}
This is superpolynomially small.
\end{proof}

\begin{remark} \label{rem:qec}
Finally, we briefly discuss the effect of Quantum Error Correction.
The results in this paper can be viewed in a practical light as giving
a necessary bound to the error tolerance any Quantum Error Correction must provide
in order to, say, factor large integers. It is true that
the Threshold Theorem~\cite{aharonov_fault_2008, knill_resilient_1998} implies that if one can use a large number of
physical qubits to represent a single qubit, then one can correct (essentially) independent noise (such as those modeled in this paper) if one can keep 
physical noise level below a certain constant. However, one must also
realize that this theorem is proved under the  assumption that a system with
a large number of qubits  is  modeled, with infinite accuracy, as a tensor
product space. This latter assumption is a mathematical idealization for
which experimental evidence is lacking \cite{alicki_internal_2006,alicki_quantum_2006}. The first author 
believes (as stated in~\cite{cai}) that the exact SU(2) group description
of permissible single qubit rotations is only a mathematical model,
and that rotations of angles  below some sufficiently small threshold 
lose physical meaning. By the same token he does not believe
the tensor product model for a large number of
physical qubits is infinitely accurate.
Thus, he does not believe the difficulty discussed in this paper
is solved by the Threshold Theorem, and this  is not merely
a matter of engineering implementation.
\end{remark}

\section*{Acknowledgements}

Ben Young is supported by the European Union (CountHom, 101077083). Views and opinions expressed are however those of the author(s) only and do not necessarily reflect those of the European Union or the European Research Council Executive Agency. Neither the European Union nor the granting authority can be held responsible for them. Ben Young also acknowledges the support of
a Cisco Fellowship at the University of Wisconsin-Madison.

\printbibliography
\end{document}